\documentclass[11pt,onecolumn]{article}

\usepackage[T1]{fontenc}
\usepackage[utf8]{inputenc}
\usepackage{lmodern}
\usepackage{microtype}
\usepackage{amsmath,amssymb,amsthm,mathtools}
\usepackage{bm}
\usepackage{graphicx}
\usepackage{float}
\usepackage{xcolor}
\usepackage{tikz}
\usetikzlibrary{shapes,arrows,positioning,calc}
\usepackage{booktabs}
\usepackage{enumitem}
\usepackage[margin=1.25in]{geometry}
\usepackage{setspace}
\usepackage{xspace}
\usepackage{url}
\usepackage[colorlinks=true,allcolors=blue]{hyperref}
\usepackage[nameinlink,noabbrev]{cleveref}

\newtheorem{theorem}{Theorem}
\newtheorem{lemma}[theorem]{Lemma}

\newtheorem{corollary}[theorem]{Corollary}
\theoremstyle{definition}

\newcommand{\br}[1]{\mathopen{}\left( #1 \right)}
\newcommand{\brc}[1]{\mathopen{}\left\{ #1 \right\}}
\newcommand{\spr}[1]{\mathopen{}\left| #1 \right|}

\newcommand{\cov}{\operatorname{cov}}

\newcommand{\NP}{\textnormal{NP}\xspace}
\newcommand{\Pclass}{\textnormal{P}\xspace}

\newcommand{\OPT}{\text{OPT}}
\newcommand{\pathbt}{\text{path}}

\newcommand{\cost}{\text{cost}}
\newcommand{\COSTW}{\text{cost}^W}
\newcommand{\COSTA}{\text{cost}^A}

\newcommand{\argmax}{\mathopen{}\operatorname*{arg\,max}}

\newcommand{\cH}{\mathcal{H}}

\newcommand{\cS}{\mathcal{S}}
\newcommand{\cT}{\mathcal{T}}
\newcommand{\cC}{\mathcal{C}}
\newcommand{\cF}{\mathcal{F}}
\newcommand{\cU}{\mathcal{U}}
\newcommand{\cR}{\mathcal{R}}
\newcommand{\cI}{\mathcal{I}}
\newcommand{\cJ}{\mathcal{J}}
\newcommand{\cG}{\mathcal{G}}

\newcommand{\cO}{\mathcal{O}}

\newcommand{\cP}{\mathcal{P}}

\title{Simpler Logarithmic Approximation Algorithms for the Optimal Decision Tree and Adaptive Set Cover}
\author{
  Michał Szyfelbein
}
\date{\today}

\begin{document}
\maketitle

\begin{abstract}
We study the problem of constructing a decision tree that identifies an unknown hypothesis from a given ground set, under both average- and worst-case cost. The \textsc{Optimal Decision Tree} (ODT) problem has been extensively studied in the literature, and $\mathcal{O}(\log n)$-approximation guarantees are known for each criterion, including a simultaneous guarantee for both (here $n$ is the number of hypotheses) [ICML]. Although most of the existing algorithm use simple greedy strategies, their analyses rely on intricate arguments. Since ODT is a fundamental problem with wide-ranging applications, we are interested in developing algorithms that are both effective and easy to understand.

We present a new greedy algorithm with a simplified analysis that simultaneously achieves approximation factors below $3.65\log_2 n$ for worst-case cost and $7.3\log_2 n$ for average-case cost.
Our algorithm for the \textsc{Optimal Decision Tree} behaves greedily with respect to the hereby introduced \textsc{Separating Subfamily} problem which asks for the cheapest subfamily of tests which partitions the hypotheses into pieces of small enough size. We show that the \textsc{Separating Subfamily} can itself be reduced to an instance of the well-known \textsc{Maximum Coverage} problem. At the core of our approach lies exploiting properties of cutting a clique into small pieces, where edges represent pairs of hypotheses to be separated. In the spirit of simplicity, we also provide a new and straightforward application of ODT: an $\mathcal{O}(\log (k\cdot n))$-approximation for the \textsc{Adaptive Set Cover} problem, where $k$ is the number of possible realizations.
\end{abstract}


\section{Introduction}
The \textsc{Optimal Decision Tree} is a central problem in computer science, laying at the intersection of combinatorial optimization and machine learning. 
Consider the task of identifying an unknown target hypothesis $h^*$ from an a-priori known, ground set of $n$ hypotheses~$\cH$, by performing a sequence of tests from a set $\cT$ of size $m$. Each test $t\in \cT$ is a partition of hypotheses into subsets and has an associated cost $c\br{t}$. Performing a test $t$ incurs a cost of~$c\br{t}$ and reveals the part of $t$ containing $h^*$. The hidden hypothesis $h^*$ is identified whenever it is the only remaining one consistent with the outcomes of the already executed tests. Since these tests are performed adaptively, any test execution strategy can be encoded as a decision tree, where internal nodes correspond to tests, edges represent possible responses, and leaves correspond to possible hypotheses. The task of constructing such a tree is known as the \textsc{Optimal Decision Tree} (ODT) problem and has been extensively studied in the literature. The two most commonly considered quality measures of a decision tree are the worst-case and average-case depth, corresponding to the maximum and expected cost of identifying a hypothesis, respectively. A simultaneous $\cO\br{\log n}$-approximation for both objectives is already known~\cite{DiagnosisDeterminationDecisionTreesOptimizingSimultaneouslyWorstAndExpectedTestingCost}. However, the analysis of this algorithm is very involved, and since ODT is a fundamental problem in computer science, we ask whether an easily understanable approach can establish similar result. Hereby, we answer this question affirmatively, and give a simple greedy algorithm with simultaneous, relateively low approximation factors of $3.65\log_2 n$ for worst-case cost and $7.3\log_2 n$ for average-case cost, together with an intuitive analysis.

Most heuristics for the ODT problem behave greedily with respect to some naturally defined subproblem. For example, they may select the test that separates the hypotheses most evenly according to a given measure of evenness. Importantly, these strategies are usually concerned with choosing a single test at a time. Since there may be no test that significantly reduces the size of the instance, proving that such a rule achieves a nontrivial approximation ratio often requires careful analysis. To circumvent this inconvenience, we introduce a greedy strategy that selects a subfamily of tests rather than a single one. The intuition behind the performance guarantee of our algorithm can be summarized in a single sentence: we pay $O\br{\OPT}$ to split the hypotheses into relatively small subsets and recurse on each subset. This leads to a logarithmic number of recursion levels and, consequently, an $O\br{\log n}$-approximation ratio.

The main technical contribution required to obtain our result is a constant-factor approximation algorithm for the \textsc{Separating Subfamily}, which is a natural subproblem for the ODT. To do so, we use a further reduction to the well-known \textsc{Maximum Coverage} problem. Our construction crucially exploits the properties of partitioning a clique into small pieces, whose edges represent pairs of hypotheses that need to be separated. Luckily, even here the analysis is relatively straightforward and comprehensive. As an additional result, we show a new application of the ODT in the spririt of simplcity, by using it to design an $\cO\br{\log \br{k\cdot n}}$-approximation algorithm for the \textsc{Adaptive Set Cover} problem.

\subsection{Related Work}

We list several most important contributions concerning the ODT as well as \textsc{Set Cover} and \textsc{Adaptive Submodularity}, two closely related areas.

\paragraph*{Optimal Decision Tree}
Decision tree construction belongs to the most well-known topics in computer science. The problem has been widely studied since 1970s starting with Garey and Graham \cite{OptimalBinaryIdentificationProcedures,PerformanceBoundsOnTheSplittingAlgorithmForBinaryTesting}, and since then it has been almost completely characterized with regards to both quality measures. The concrete variant depends on whether the tests are binary, the costs are uniform, and when the average-case cost is studied, whether the probability distribution over the hypotheses is uniform. All of these setups are known to be NP-hard \cite{OptimalBinaryIdentificationProcedures,ConstructingOptimalBinaryDecisionTreesIsNPComplete}. Moreover, they cannot be approximated within an $o\br{\log n}$ factor \cite{OnTheHardnessOfTheMinimumHeightDecisionTreeProblem,DecisionTreesForEntityIdentificationApproximationAlgorithmsAndHardnessResults, DiagnosisDeterminationDecisionTreesOptimizingSimultaneouslyWorstAndExpectedTestingCost}, except the average-case cost minimization with a constant number of outcomes, uniform costs and probabilities. Multiple works regarding approximating various special cases of the problem have arisen. For all of these settings, an $\cO\br{\log n}$-approximation algorithms exist, including \cite{OnAnOptimalSplitTreeProblem, AnalysisOfAGreedyActiveLearningStrategy, ApproximatingDecisionTreesWithMultiwayBranches, Average-caseActiveLearningWithCosts, ApproximatingOptimalBinaryDecisionTrees,DecisionTreesForEntityIdentificationApproximationAlgorithmsAndHardnessResults,DiagnosisDeterminationDecisionTreesOptimizingSimultaneouslyWorstAndExpectedTestingCost, ApproximationAlgorithmsForOptimalDecisionTreesAndAdaptiveTSPProblems, AdaptiveSubmodularRankingAndRouting, AdaptivityInAdaptiveSubmodularity, MinimumCostAdaptiveSubmodularCover, ASimpleApproximationAlgorithmForOptimalDecisionTree} for the average-case and \cite{DecisionTreesForGeometricModels,DiagnosisDeterminationDecisionTreesOptimizingSimultaneouslyWorstAndExpectedTestingCost,TheCostComplexityOfInteractiveLearning} for the worst-case. One can always combine the algorithms for the average- and worst-case cost to achieve asymptotically the same approximation guarantees for both quality measures \cite{TradingOffWorstAndExpectedCostInDecisionTreeProblems}. For average-case ODT with uniform probabilities and costs and a fixed number of responses, Li, Liang, and Mussmann analyze a particular greedy algorithm and show that it achieves approximation of $\cO\br{\log n/\log\log n}$ \cite{ATightAnalysisOfGreedyYieldsSubexponentialTimeApproximationForUniformDecisionTree}. Their work also gives a subexponential-time, constant-factor approximation algorithm, which also implies that any superconstant guarantee is not NP-hard to achieve under the \textsc{Exponential Time Hypothesis}.

Recently, the problem has also been generalized to a setting where a partial ordering on the tests is given, imposing precedence constraints among tests, and the goal is to find a decision tree that respects this ordering \cite{Precedence-ConstrainedDecisionTreesAndCoverings}. Depending on the type of the partial ordering the approximability changes drastically. For inforests, an $\cO\br{\log n}$-approximation is possible. For outforests an $\cO\br{\log^2n}$-approximation algorirthm exists, which for the worst-case criterion is best possible unless $NP\subseteq ZTIME\br{n^{\text{poly}\br{\log n}}}$. For general POSETs the approximation ratio given is $\cO^*\br{m^{1/2}}$, contrasted by inapproximability factor of $o\br{m^{1/12-\epsilon}}$ conditioned on the \textsc{Planted Dense Subgraph Conjecture}. These results hold for both cost criterions, any number of outcomes and uniform costs and probabilities.

\paragraph*{Set Cover}
The \textsc{Set Cover} is a classical combinatorial optimization problem, and it is well known that it admits an $H_n$-approximation algorithm \cite{AGreedyHeuristicForTheSetCoveringProblem}. Moreover, this guarantee is essentially tight, since achieving an $\br{1-o\br{1}}\cdot\ln n$-factor guarantee is impossible unless $\textup{P}=\textup{NP}$~\cite{AnalyticalApproachToParallelRepetition}. There exists also a setup called \textsc{Maximum Coverage} where an additional budget constraint is given and the goal is to maximize the number of covered elements, within this budget. This task admits an $1-1/e$-approximation algorithm which is best possible unless $\textup{P}=\textup{NP}$ \cite{TheBudgetedMaximumCoverageProblem,AnalyticalApproachToParallelRepetition}. All those setups can be generalized by the \textsc{Submodular Optimization} framework, where the goal is minimize/maximize a submodular function subject to some constraints. Submodular functions capture the law of diminishing returns, i.~e., the marginal gain of adding an element to a set never increases as it grows. The \textsc{Maximum Coverage} problem can be generalized by the \textsc{Submodular Optimization under Knapsack Constraint} \cite{ANoteOnMaximizingASubmodularSetFunctionSubjectToAKnapsackConstraint}. The \textsc{Set Cover} can be generalized by the \textsc{Submodular Cover} \cite{AnAnalysisOfTheGreedyAlgorithmForTheSubmodularSetCoveringProblem}. The list of all contributions regarding submodular optimization lays beyond scope of this work. Another important variant is the \textsc{Min-Sum Set Cover}, where one minimizes the sum of covering times of elements, which has a $4$-approximation algorithm, and cannot be approximated within $(4-\epsilon)$-factor unless $\textup{P}=\textup{NP}$ \cite{ApproximatingMinSumSetCover}, for any $\epsilon>0$.

\paragraph*{Adaptive Submodularity}
The \textsc{Adaptive Submodularity}~\cite{AdaptiveSubmodularityANewApproachToActiveLearningAndStochasticOptimization} can be viewed as an adaptive/stochastic generalization of the \textsc{Submodular Cover}. In this framework, one optimizes a submodular function over partial realizations revealed by sequential testing. Both ODT and covering can be expressed as adaptive covering tasks, unifying these problems under a common framework. The work on \textsc{Adaptive Submodular Ranking}~\cite{AdaptiveSubmodularRankingAndRouting} provides logarithmic guarantees and yields, in particular, an $\cO\br{\log n}$ approximation for the average-case ODT through a reduction. Moreover, results on the \textsc{Adaptive Cover} such as~\cite{AdaptivityInAdaptiveSubmodularity,MinimumCostAdaptiveSubmodularCover} give improved guarantees for the average-case ODT (e.g., $2\cdot \ln\br{1/p_{\min}}$ for uniform costs and $4\cdot\ln n$ for uniform probabilities), while~\cite{MinimumCostAdaptiveSubmodularCover} also proves a general $4\cdot \br{1+\ln Q}$ bound for the \textsc{Adaptive-Submodular Cover} (for a goal value of $Q$).
\section{Problem Setup}\label{sec:problem_setup}
We are given a set $\cH$ of $n$ \emph{hypotheses} and a set $\cT$ of $m$ \emph{tests}, associated with a \emph{cost function} $c\colon \cT\to \mathbb{N}$. Each test $t\in\cT$ is a partition of $\cH$ into disjoint subsets called \emph{responses}. Among $\cH$ there is a \emph{hidden hypothesis} $h^*$ which is required to be identified. Performing a test $t\in \cT$ incurs a cost of $c\br{t}$ and reveals the partition of $t$ containing $h^*$.
A test $t\in\cT$ is said to \emph{separate} a pair of hypotheses $h, h'\in \cH$ if they belong to different parts of $t$. We assume that for each such pair, there exists a test $t\in \cT$ separating them. 
A decision tree $D$ is a rooted tree where each internal node is labeled with a test $t\in \cT$ and each edge from an internal node to its child is encodes a possible response to the corresponding test. Each leaf of the tree is labeled with a hypothesis $h\in \cH$ and corresponds to the fact that $h$ was identified to be $h^*$, i.~e., separated from the remainder of $\cH$. Let $\pathbt_D\br{v}$ denote the set of tests in the unique path from the root to a vertex $v$. The \emph{identification cost} of a vertex $v$ in $D$, is defined as $\cost\br{D, v}=\sum_{t\in \pathbt_D\br{v}}c\br{t}$. The \emph{average-case identification cost}\footnote{
Note, that for the sake of convenience we do not divide this measure by $n$.} of $D$ is defined as 

\[\COSTA\br{D}=\sum_{h\in \cH} \cost\br{D, h}.\]
The \emph{worst-case identification cost} of $D$ is defined as 
\[
\COSTW\br{D}=\max_{h\in \cH} \brc{\cost\br{D, h}}.
\]
The goal of the \textsc{Optimal Decision Tree} (ODT) problem is to find a decision tree $D$ that minimizes either $\COSTA$ or $\COSTW$. We denote an instance of this problem by $\cI=\br{\cH,\cT,c}$. We use $\OPT^W(\cI)$ and $\OPT^A(\cI)$ to denote the optimal worst-case and average-case costs of $\cI$, respectively. For $\cH'\subseteq\cH$, we write $\cI_{\cH'}=\br{\cH',\cT,c}$.

 As a key subproblem we will require the \textsc{Separating Subfamily} problem: We are given a universe $\cU$ of $n$ elements, a family of $m$ partitions $\cT$ of $\cU$ with a cost function $c\colon \cT\to \mathbb{N}$, and a real number $f\in [1/n,1]$. For a subfamily $\cF\subseteq\cT$ we write $c\br{\cF}=\sum_{\cP\in\cF} c\br{\cP}$. We say that
 a partition $\cP\in \cT$ \textit{separates} $u,v\in \cU$ if $u$ and $v$ belong to different parts of~$\cP$. Let $\cP_{\cF}$ be the partition of $\cU$ induced by~$\cF$, where two elements belong to the same part iff no member of~$\cF$ separates them. In our instance $\cI=\br{\cU, \cT, c, f}$, we demand that $\cP_{\cT}=\brc{\brc{u}\colon u\in \cU}$. The goal of the \textsc{Separating Subfamily} problem is to find a subfamily $\cF\subseteq \cT$ minimizing $c\br{\cF}$, such that for each $P\in \cP_{\cF}$, $|P|\leq f\cdot |\cU|$. Denote by $\OPT_{SF}\br{\cI}$ the optimum \textsc{Separating Subfamily} cost on $\cI$.
 For any $\alpha,\beta\geq 1$ we say that an algorithm is a bicriteria $\br{\alpha, \beta}$-approximation algorithm for the \textsc{Separating Subfamily} if it outputs a subfamily $\cF\subseteq\cT$, such that $c\br{\cF}\leq \alpha\cdot \OPT_{SF}\br{\cI}$, and for each $P\in \cP_{\cF}$, $|P|\leq \beta\cdot f\cdot |\cU|$.

We also use the \textsc{Maximum Coverage} problem. We are given a universe $\cU$ of $n$ elements, a family $\cS$ of $m$ covering subsets of $\cU$ (where $\bigcup_{S\in\cS}S=\cU$), with a cost function $c\br{S}\colon \cS\to \mathbb{N}$ and a budget $B\in \mathbb{N}$. For any $\cC\subseteq \cS$, called \emph{cover}, let $\cov\br{\cC}=\bigcup_{S\in \cC} S$ be the set of elements covered by $\cC$. The goal is to maximize $\spr{\cov\br{\cC}}$ over $\cC\subseteq \cS$, subject to $c\br{\cC}\leq B$. For an instance $\cI=\br{\cU, \cS, c, B}$, let $\OPT_{MC}\br{\cI}$ be the optimum \textsc{Maximum Coverage} value on~$\cI$.

Given a simple graph $G=\br{V, E}$, by a \emph{triangle cut} we mean any non-empty subset of edges $S\subseteq E$ which obeys the triangle inequality. We say that a set of edges $S$ satisfies the \emph{triangle inequality} if for every three vertices $u,v,w\in V$ such that $uv, uw, vw\in E$, if $uv\in S$, then either $uw\in S$ or $vw\in S$ (or both). For any $S\subseteq E$, we denote by $G\setminus S$ the set of connected components occuring after the removal of edges contained in~$S$, from $G$. For any family of edge sets $\cS$, (where $S\subseteq E$ for any $S\in\cS$), we write $E_{\cS}=\bigcup_{S\in \cS} S$.
\section{The algorithm}

We begin by showing how to use an $\br{\alpha, \beta}$-approximation algorithm for the \textsc{Separating Subfamily} to construct an $\cO\br{\alpha\cdot\log_{2/\beta} n}$-approximate decision tree for both cost criteria.
The decision tree built by our algorithm is given implicitly, since for clarity we state the procedure as a strategy of performing tests. It is straightforward to convert this strategy into an explicit decision tree, call it $D$, by performing a depth-first traversal of the recursion tree for all possible outcomes of the tests. The algorithm is as follows:
\begin{enumerate}
    \item If $\spr{\cH}=1$, then the hidden hypothesis is identified and the algorithm terminates.
    \item Otherwise, use the $\br{\alpha, \beta}$-approximation algorithm for the \textsc{Separating Subfamily} instance $\cI_{SF}=\br{\cU=\cH, \cT, c, f=1/2}$ to find an $\br{\alpha, \beta}$-approximate subfamily $\cF\subseteq \cT$ that separates $\cH$ into pieces of size at most $\beta\cdot f\cdot |\cH| = \beta\cdot |\cH|/2$.\label{step:find_family}
    \item Perform any test $t\in \cF$ on the current candidate set $\cH$ and observe the outcome~$P$ (the one which contains the hidden hypothesis $h^*$). \label{step:perform_test}
    \item If $\spr{P}\leq \beta\cdot \spr{\cH}/2$, then recurse on $\cI_{P}$.
    \item Otherwise, repeat Step \ref{step:perform_test} on $\cH \gets P$ and the set of tests $\cT\gets \cT\setminus t$.
\end{enumerate}

\begin{theorem}\label{thm:main_algorithm}
    For an instance $\cI=\br{\cH, \cT, c}$ of the \textsc{Optimal Decision Tree} problem, and the implicitly constructed decision tree $D$, we have that
    \[
    \COSTW\br{D}\leq \alpha\cdot\log_{2/\beta} n\cdot \OPT^W\br{\cI} \quad \text{and}\quad
    \COSTA\br{D}\leq 2\alpha\cdot\log_{2/\beta} n\cdot \OPT^A\br{\cI}.
    \]
\end{theorem}
\begin{proof}
    We begin by establishing lower bounds on $\OPT^W\br{\cI}$ and $\OPT^A\br{\cI}$. Let $D_W^*$ and $D_A^*$ be optimal decision trees for the worst-case and average-case costs, respectively.
    Let $H_v$ ne the set of hypotheses consistent with all previous responses at node $v$. Starting from the root of each optimal tree, repeatedly follow the child $u$ with largest $\spr{H_u}$ while $\spr{H_u}>\spr{\cH}/2$. Stop at the first node whose children all have candidate sets of size at most $\spr{\cH}/2$, and call these stopping nodes $v_W$ and $v_A$. Thus every test on either root-to-stopping-node path is performed while at least half of the hypotheses remain. In particular, $\cost\br{D_W^*, v_W}\leq\OPT^W\br{\cI}$, since it is a root-to-node path cost in the optimal worst-case tree.

    For the average-case tree, each test $t$ on the path to $v_A$ contributes $c\br{t}$ to the identification cost of at least $n/2$ hypotheses, and therefore contributes at least $c\br{t}\cdot n/2$ to $\COSTA\br{D_A^*}=\OPT^A\br{\cI}$. Summing over this path yields $\cost\br{D_A^*, v_A}\cdot n/2\leq\OPT^A\br{\cI}$.

    By definition, both $\pathbt_{D_W^*}\br{v_W}$ and $\pathbt_{D_A^*}\br{v_A}$ form separating subfamilies for the instance $\cI_{SF}$. Thus, $\OPT_{SF}\br{\cI_{SF}}\leq \cost\br{D_W^*, v_W}$ and $\OPT_{SF}\br{\cI_{SF}}\leq \cost\br{D_A^*, v_A}$. Since we use an $\br{\alpha, \beta}$-approximation algorithm for the \textsc{Separating Subfamily}, we have that $c\br{\cF}\leq \alpha\cdot \cost\br{D_W^*, v_W}$ and $c\br{\cF}\leq \alpha\cdot \cost\br{D_A^*, v_A}$. Moreover, we know that for every $P\in \cP_{\cF}, \spr{P}\leq\beta\cdot \spr{\cH}/2$, and therefore, for each possible target hypothesis the algorithm eventually recurses.

    We also make use of the following two immediate folklore observations. Let $\cP$ be any partition of $\cH$ (not necessarily $\cP_{\cF}$ or a member of $\cT$). We have that
    \[
    \max_{P\in \cP} \brc{\OPT^W\br{\cI_P}}\leq \OPT^W\br{\cI}\quad \text{and}\quad
    \sum_{P\in \cP} \OPT^A\br{\cI_P}\leq \OPT^A\br{\cI}.
    \]

    The first inequality follows because restricting any decision tree for $\cI$ to hypotheses in a fixed $P$ gives a decision tree for $\cI_P$ with no greater worst-case cost. The second inequality follows from the fact that for any $P\in \cP$, any decision tree $D$ for $\cH$ can also be used to identify the target hypothesis in $P$, and members of $P\in \cP$ are disjoint. The contribution of any test $t\in \cT$ in $D$ is the size of the set of hypotheses it is performed in. This same set is partitioned by $\cP$ into disjoint subsets, and therefore, the sum of the contributions of this test over all $P\in \cP$ is at most the contribution of this test in $D$.

    Armed with the above observations, we prove the following claim by induction on $\spr{\cH}$:
    \[
    \COSTW\br{D}\leq \alpha\cdot \log_{2/\beta} n\cdot \OPT^W\br{\cI}\quad \text{and} \quad \COSTA\br{D}\leq 2\alpha\cdot \log_{2/\beta} n\cdot \OPT^A\br{\cI}.
    \]
    For the base case, when we have $\spr{\cH}=1$, the claim holds trivially since no tests are present in $D$ and $\COSTW\br{D}=\COSTA\br{D}=0$.

    Now, assume by induction that the claim holds for all $\cH$, such that $\spr{\cH}<k$. Let $\spr{\cH}=k$ and let $\cP$ be the partition of $\cH$ consisting of all of the responses which might occur when performing the tests in $\cF$ as in the algorithm\footnote{Here, observe that $\cP$ might not be exactly $\cP_{\cF}$. Some responses in  $\cP$ might be a union of members of $\cP_{\cF}$.}. Moreover, for each $P\in \cP$, let $D_P$ be the subtree of $D$ associated with $P$, constructed recursively.
    We have that
    \begin{align*}
        \COSTW\br{D}&\leq c\br{\cF}+\max_{P\in \cP} \brc{\COSTW\br{D_P}}
        \\&\leq \alpha\cdot \cost\br{D_W^*, v_W}+\max_{P\in \cP} \brc{\COSTW\br{D_P}}
        \\&\leq \alpha\cdot \OPT^W\br{\cI}+\max_{P\in \cP} \brc{\COSTW\br{D_P}}\\
        &\leq \alpha\cdot \OPT^W\br{\cI}+\max_{P\in \cP} \brc{\alpha\cdot \log_{2/\beta} |P|\cdot \OPT^W\br{\cI_P}}\\
        &\leq \alpha\cdot \OPT^W\br{\cI}+\alpha\cdot \log_{2/\beta} \frac{\beta\cdot \spr{\cH}}{2}\cdot \max_{P\in \cP}\brc{\OPT^W\br{\cI_P}}\\
        &\leq \alpha\cdot \OPT^W\br{\cI}+\alpha\cdot \log_{2/\beta} \frac{\beta\cdot \spr{\cH}}{2}\cdot \OPT^W\br{\cI}= \alpha\cdot \log_{2/\beta} k\cdot \OPT^W\br{\cI},
    \end{align*}
    where the second inequality follows since $\cF$ is an $\br{\alpha, \beta}$-approximation to the \textsc{Separating Subfamily} problem, and the third inequality follows from the lower bound $\cost\br{D_W^*, v_W}\leq \OPT^W\br{\cI}$. The fourth inequality is by the inductive hypothesis, the fifth inequality comes from the fact that for $P\in \cP$, $\spr{P}\leq \beta\cdot \spr{\cH}/2$, and the sixth inequality is by $\max_{P\in \cP}\brc{\OPT^W\br{\cI_P}}\leq \OPT^W\br{\cI}$.

    Similarly, we have that
    \begin{align*}
        \COSTA\br{D}&\leq c\br{\cF}\cdot n+\sum_{P\in \cP} \COSTA\br{D_P}
        \\&\leq \alpha\cdot \cost\br{D_A^*, v_A}\cdot n+\sum_{P\in \cP} \COSTA\br{D_P}\\
        &\leq 2\alpha\cdot \OPT^A\br{\cI}+\sum_{P\in \cP} \COSTA\br{D_P}\\&\leq 2\alpha\cdot \OPT^A\br{\cI}+\sum_{P\in \cP} 2\alpha\cdot \log_{2/\beta} |P|\cdot \OPT^A\br{\cI_P}
        \\&\leq 2\alpha\cdot \OPT^A\br{\cI}+2\alpha\cdot \log_{2/\beta} \frac{\beta\cdot \spr{\cH}}{2}\cdot \sum_{P\in \cP}\OPT^A\br{\cI_P}\\
        &\leq 2\alpha\cdot \OPT^A\br{\cI}+2\alpha\cdot \log_{2/\beta} \frac{\beta\cdot \spr{\cH}}{2}\cdot \OPT^A\br{\cI} = 2\alpha\cdot \log_{2/\beta} k\cdot \OPT^A\br{\cI}.
    \end{align*}
    where the second inequality follows from the fact that $\cF$ is an $\br{\alpha, \beta}$-approximation to the \textsc{Separating Subfamily}, and the third one from the lower bound $\cost\br{D_A^*, v_A}\cdot n/2\leq \OPT^A\br{\cI}$. The fourth inequality follows from the inductive hypothesis, the fifth since for $P\in \cP$, $\spr{P}\leq \beta\cdot \spr{\cH}/2$, and the sixth from $\sum_{P\in \cP}\OPT^A\br{\cI_P}\leq \OPT^A\br{\cI}$.
\end{proof}
\section{Finding a Good Separating Subfamily}
In this section we present a bicriteria, constant-factor approximation algorithm for the \textsc{Separating Subfamily} problem.
To do so, we use the following well-known result of \cite{TheBudgetedMaximumCoverageProblem} regarding approximating the \textsc{Maximum Coverage} problem:
\begin{theorem}
  There is an $\br{1-1/e}$-approximation algorithm for the \textsc{Maximum Coverage}.
\end{theorem}
We prove the following theorem, which is the main result of this section:
\begin{theorem}\label{thm:separating_subfamily}
There is a polynomial-time bicriteria $\br{1, \frac{\sqrt{1-\br{1-1/e}\cdot\br{1-f}}}{f}}$-approximation algorithm for the \textsc{Separating Subfamily}.
\end{theorem}
\begin{proof}
Let $\cI=\br{\cU, \cT, c, f}$ be a \textsc{Separating Subfamily} instance. For any partition $\cP$ of~$\cU$, denote 
\[
S_{\cP}=\brc{\{u,v\}\in \binom{\cU}{2}\colon u,v\text{ are separated by }\cP}.
\]
For each integer budget $B$, let $\cI_{MC}^{B}$ be the \textsc{Maximum Coverage} instance on universe $\binom{\cU}{2}$ with sets $\cS=\{S_{\cP}\colon \cP\in\cT\}$, costs $c\br{S_{\cP}}=c\br{\cP}$, and budget $B$. Set $\alpha=1-1/e$ and $T=\frac{1-f}{2}\cdot\spr{\cU}^2$.
The algorithm is as follows:
\begin{enumerate}
  \item Let $C=\sum_{\cP\in\cT}c\br{\cP}$. Run the $\alpha$-approximation algorithm on $\cI_{MC}^{C}$ and let $\widehat{\cF}$ be the corresponding subfamily.
  \item Binary-search the integer budgets from $0$ to $C$. At each queried budget $B$, run the $\alpha$-approximation algorithm on $\cI_{MC}^{B}$ to obtain a cover $\cC_B$, and let $\cF_B=\{\cP\in\cT\colon S_{\cP}\in\cC_B\}$. If $\spr{\cov\br{\cC_B}}\geq\alpha T$, retain $\cF_B$ and continue searching smaller budgets. Otherwise, continue searching larger budgets. Return the retained subfamily with the lowest budget found by the search.
\end{enumerate}

The binary search makes $\cO\br{\log\br{C}}$ calls to the polynomial-time \textsc{Maximum Coverage} approximation algorithm, so it runs in polynomial time in the input size.
To prove that all of the pieces are of small enough size, we begin with two lemmas regarding cutting a complete graph into small pieces. Intuitively, we need these statements to show what happens to the size of members of the partition $P_{\cF}$, when cutting edges representing pairs of elements in $\cU$ that are separated by the returned subfamily $\cF$.
\begin{lemma}\label{lemma:cutting_clique_lower_bound}
  Let $K_n$ be a complete graph on $n$ vertices, $f\in [1/n,1]$ be a real number and $S\subseteq E\br{K_n}$. If for every connected component $H\in K_n\setminus S$, $\spr{V\br{H}}\leq f\cdot n$, then $|S|\geq \frac{1-f}{2}\cdot n^2$.
\end{lemma}
\begin{proof}
  Let $|S|=k$, and assume that sizes of members of $K_n\setminus S$ are $n_1,\dots,n_p$, with $n_i \le f\cdot n$ for $i\in[p]$ and $\sum_{i=1}^p n_i = n$. By counting the edges between members of $K_n\setminus S$ in $K_n$ we get
  \[
  k \geq \binom{n}{2}-\sum_{i=1}^p\binom{n_i}{2} = \frac{1}{2} \cdot\br{n^2-n-\sum_{i=1}^p n_i^2-\sum_{i=1}^pn_i} =\frac{1}{2} \cdot\left(n^2 - \sum_{i=1}^p n_i^2\right).
  \]
Since for every $i\in [p]$, $n_i\leq f\cdot n$, we obtain
\[
\sum_{i=1}^p n_i^2 \leq f\cdot n \cdot \sum_{i=1}^p n_i = f\cdot n^2,
\]
Therefore,
\[
k \geq \frac{1}{2}\cdot\br{n^2 - f\cdot n^2} = \frac{1-f}{2}\cdot n^2,
\]
which gives the claim.
\end{proof}
\begin{lemma}\label{lemma:cutting_clique_upper_bound}
  Let $K_n$ be a complete graph on $n$ vertices, $f\in [1/n, 1]$ be a real number and $\cJ$ be a set of triangle cuts of $K_n$. If $|E_{\cJ}|\geq \gamma \cdot \frac{1-f}{2}\cdot n^2$ for some constant $\gamma\in [0,1]$, then for each $H\in K_n\setminus E_{\cJ}$, $\spr{V\br{H}}\leq \sqrt{1-\gamma\cdot \br{1-f}}\cdot n$.
\end{lemma}
\begin{figure}[H]
  \centering
  \colorlet{Kblob}{gray!10!white}
  \colorlet{Sblob}{blue!12!white}
  \colorlet{Lblob}{red!12!white}
  \colorlet{Snode}{blue!55!white}
  \colorlet{Lnode}{red!55!white}
  \colorlet{Ledge}{red!70!black}

  \begin{tikzpicture}[scale=1.0]
    \draw[fill=Kblob, draw=gray!65!black, thick, dashed] (0.0, 0.1) ellipse (6.2cm and 2.55cm);
    \node[font=\small\bfseries, gray!70!black] at (0.0, 2.35) {$K_n$};

    \node[ellipse, draw=blue!60!black, fill=Sblob, thick, minimum width=3.9cm, minimum height=2.7cm] (SblobNode) at (-3.60, 0.05) {};
    \node[ellipse, draw=red!65!black, fill=Lblob, thick, minimum width=5.4cm, minimum height=3.8cm] (LblobNode) at (3.05, 0.05) {};
    \node[font=\small\bfseries, blue!60!black] at (-3.60, 1.72) {$S$};
    \node[font=\small\bfseries, red!65!black] at (2.8, 2.15) {$L$};

    \foreach \i in {1,...,8} {
      \coordinate (s\i) at ($(-3.60,0.05)+(45*\i:1.05)$);
    }

    \foreach \i in {1,...,14} {
      \coordinate (l\i) at ($(3.05,0.05)+(12+\i*360/14:1.64)$);
    }

    \foreach \i in {1,...,8} {
      \foreach \j in {1,...,8} {
        \ifnum\i<\j
          \draw[blue!55!black, dashed, line width=0.55pt, opacity=0.55] (s\i)--(s\j);
        \fi
      }
    }

    \foreach \i in {1,...,14} {
      \foreach \j in {1,...,14} {
        \ifnum\i<\j
          \draw[Ledge, line width=0.75pt, opacity=0.75] (l\i)--(l\j);
        \fi
      }
    }

    \draw[gray!60!black, dashed, line width=0.95pt, opacity=0.78] (SblobNode.270) -- (LblobNode.212);
    \draw[gray!60!black, dashed, line width=0.95pt, opacity=0.78] (SblobNode.49)  -- (LblobNode.128);
    \draw[gray!60!black, dashed, line width=0.95pt, opacity=0.78] (SblobNode.312) -- (LblobNode.196);
    \draw[gray!60!black, dashed, line width=0.95pt, opacity=0.78] (SblobNode.35)  -- (LblobNode.158);
    \draw[gray!60!black, dashed, line width=0.95pt, opacity=0.78] (SblobNode.346) -- (LblobNode.232);
    \draw[gray!60!black, dashed, line width=0.95pt, opacity=0.78] (SblobNode.19)  -- (LblobNode.176);
    \draw[gray!60!black, dashed, line width=0.95pt, opacity=0.78] (SblobNode.11)  -- (LblobNode.162);
    \draw[gray!60!black, dashed, line width=0.95pt, opacity=0.78] (SblobNode.3)   -- (LblobNode.202);
    \draw[gray!60!black, dashed, line width=0.95pt, opacity=0.78] (SblobNode.343) -- (LblobNode.162);
    \draw[gray!60!black, dashed, line width=0.95pt, opacity=0.78] (SblobNode.346) -- (LblobNode.140);
    \draw[gray!60!black, dashed, line width=0.95pt, opacity=0.78] (SblobNode.336) -- (LblobNode.188);

    \foreach \i in {1,...,8} {
      \filldraw[fill=Snode, draw=blue!70!black, line width=0.55pt] (s\i) circle (0.10);
    }
    \foreach \i in {1,...,14} {
      \filldraw[fill=Lnode, draw=red!75!black, line width=0.55pt] (l\i) circle (0.10);
    }
  \end{tikzpicture}

  \caption{The largest surviving component $L$ is a clique, so no edge inside $L$ is cut. With $S=V(K_n)\setminus V(L)$ and $k=|S|$, every cut edge touches $S$, giving at most $k(n-k)+\binom{k}{2}$ cut edges.}
  \label{fig:lemma4_cutting_clique}
\end{figure}
\begin{proof}
  Let $\cG=K_n\setminus E_{\cJ}$ and $L=\argmax_{H\in \cG} \brc{\spr{V\br{H}}}$, with $l=\spr{V\br{L}}$. Observe that every connected component of $\cG$ is a clique. If not we would contradict the triangle inequality. 
  Let $S=V\br{K_n}\setminus V\br{L}$ and $k=\spr{S}=n-l$.
  Since $L$ is a clique, no edge inside $V\br{L}$ belongs to $E_{\cJ}$, so every edge of $E_{\cJ}$ has an endpoint in $S$. There are exactly $k\cdot (n-k)$ edges between $S$ and $V\br{L}$ and $\binom{k}{2}$ edges inside $S$. See Figure \ref{fig:lemma4_cutting_clique} for a visual example. Therefore
  \[
  \gamma \cdot \frac{1-f}{2}\cdot n^2\leq \spr{E_{\cJ}}\leq k\cdot(n-k)+\binom{k}{2}\leq k\cdot n-\frac{k^2}{2}.
  \]
  Rearranging gives $k^2-2\cdot n\cdot k+\gamma\cdot \br{1-f}\cdot n^2\leq 0$. Solving this quadratic inequality with respect to $k$ and discarding the larger root gives
  $
  k \geq n-\sqrt{1-\gamma\cdot\br{1-f}}\cdot n.
  $
  Hence 
  \[
  l=n-k\leq \sqrt{1-\gamma\cdot\br{1-f}}\cdot n.
  \]
  Since $L$ is the largest member of $\cG$, for every $H \in \cG$, $\spr{V\br{H}}\leq l\leq \sqrt{1-\gamma\cdot\br{1-f}}\cdot n$.
\end{proof}
\begin{lemma}
For any $\gamma\in[0,1]$ and $\cF\subseteq\cT$, if $\spr{\cov\br{\{S_{\cP}\colon\cP\in\cF\}}}\geq\gamma\cdot\frac{1-f}{2}\cdot n^2$, then for each $P\in\cP_{\cF}$, $\spr{P}\leq\sqrt{1-\gamma\cdot\br{1-f}}\cdot n$.
\end{lemma}
\begin{proof}

\begin{figure}[H]
  \centering
  \colorlet{partA}{blue!14!white}
  \colorlet{partB}{red!14!white}
  \colorlet{partC}{green!16!white}
  \colorlet{nodeA}{blue!42!white}
  \colorlet{nodeB}{red!40!white}
  \colorlet{nodeC}{green!38!white}
  \colorlet{edgeA}{blue!65!black}
  \colorlet{edgeB}{red!65!black}
  \colorlet{edgeC}{green!50!black}

  \begin{minipage}[t]{0.48\textwidth}
    \centering
    \begin{tikzpicture}[scale=0.95]
      \draw[fill=partA, draw=edgeA, thick] (-2.4, 2.3) ellipse (1.5cm and 1.1cm);
      \node at (-2.4, 3.75) {\tiny\textbf{Part 1 }};
      \filldraw[fill=nodeA, draw=edgeA] (-3.2, 2.82) circle (0.11);
      \filldraw[fill=nodeA, draw=edgeA] (-2.5, 2.97) circle (0.11);
      \filldraw[fill=nodeA, draw=edgeA] (-1.8, 2.72) circle (0.11);
      \filldraw[fill=nodeA, draw=edgeA] (-3.0, 1.97) circle (0.11);
      \filldraw[fill=nodeA, draw=edgeA] (-2.3, 1.72) circle (0.11);
      \filldraw[fill=nodeA, draw=edgeA] (-1.7, 2.07) circle (0.11);

      \draw[fill=partB, draw=edgeB, thick] (-3.1, 0.20) ellipse (0.95cm and 0.72cm);
      \node at (-3.1, -0.72) {\tiny\textbf{Part 2 }};
      \filldraw[fill=nodeB, draw=edgeB] (-3.55, 0.42) circle (0.105);
      \filldraw[fill=nodeB, draw=edgeB] (-2.95, 0.52) circle (0.105);
      \filldraw[fill=nodeB, draw=edgeB] (-3.35, -0.08) circle (0.105);
      \filldraw[fill=nodeB, draw=edgeB] (-2.75, 0.02) circle (0.105);

      \draw[fill=partC, draw=edgeC, thick] (-1.05, 0.20) ellipse (0.78cm and 0.62cm);
      \node at (-1.05, -0.62) {\tiny\textbf{Part 3 }};
      \filldraw[fill=nodeC, draw=edgeC] (-1.30, 0.36) circle (0.105);
      \filldraw[fill=nodeC, draw=edgeC] (-0.80, 0.04) circle (0.105);

      \node[anchor=north, font=\tiny] at (-2.1, -1.20) {(a) Partitioned elements};
    \end{tikzpicture}
  \end{minipage}\hfill
  \begin{minipage}[t]{0.48\textwidth}
    \centering
    \begin{tikzpicture}[scale=1.00]
      \coordinate (u1) at (0, 2.2);
      \coordinate (u2) at (1.1, 1.905);
      \coordinate (u3) at (1.905, 1.1);
      \coordinate (u4) at (2.2, 0);
      \coordinate (u5) at (1.905, -1.1);
      \coordinate (u6) at (1.1, -1.905);

      \coordinate (v1) at (0, -2.2);
      \coordinate (v2) at (-1.1, -1.905);
      \coordinate (v3) at (-1.905, -1.1);
      \coordinate (v4) at (-2.2, 0);

      \coordinate (w1) at (-1.905, 1.1);
      \coordinate (w2) at (-1.1, 1.905);

      \draw[gray, very thin, opacity=0.3] (0, 0) circle (2.2);

      \draw[edgeA, line width=1.15pt]
        (u1)--(u2) (u1)--(u3) (u1)--(u4) (u1)--(u5) (u1)--(u6)
        (u2)--(u3) (u2)--(u4) (u2)--(u5) (u2)--(u6)
        (u3)--(u4) (u3)--(u5) (u3)--(u6)
        (u4)--(u5) (u4)--(u6)
        (u5)--(u6);
      \draw[edgeB, line width=1.15pt]
        (v1)--(v2) (v1)--(v3) (v1)--(v4)
        (v2)--(v3) (v2)--(v4)
        (v3)--(v4);
      \draw[edgeC, line width=1.15pt] (w1)--(w2);

      \foreach \u in {u1,u2,u3,u4,u5,u6} {
        \foreach \v in {v1,v2,v3,v4} {
          \draw[black, dashed, line width=0.75pt, opacity=0.6] (\u)--(\v);
        }
      }
      \foreach \u in {u1,u2,u3,u4,u5,u6} {
        \foreach \w in {w1,w2} {
          \draw[black, dashed, line width=0.75pt, opacity=0.6] (\u)--(\w);
        }
      }
      \foreach \v in {v1,v2,v3,v4} {
        \foreach \w in {w1,w2} {
          \draw[black, dashed, line width=0.75pt, opacity=0.6] (\v)--(\w);
        }
      }

      \foreach \u in {u1,u2,u3,u4,u5,u6} {
        \filldraw[fill=nodeA, draw=black, line width=0.45pt] (\u) circle (0.135);
      }
      \foreach \v in {v1,v2,v3,v4} {
        \filldraw[fill=nodeB, draw=black, line width=0.45pt] (\v) circle (0.135);
      }
      \foreach \w in {w1,w2} {
        \filldraw[fill=nodeC, draw=black, line width=0.45pt] (\w) circle (0.135);
      }

      \node[anchor=north, font=\tiny] at (0, -2.65) {(b) Complete graph of pairs};
    \end{tikzpicture}
  \end{minipage}

  \caption{An example partition with parts of sizes 6, 4 and 2. Figure (a) shows parts as ellipses containing elements drawn as points. Figure (b) depicts a complete graph on the points/elements, with colored intra-part edges and dashed inter-part edges.}
  \label{fig:partition_cutting}
\end{figure}

Let $K_n=\br{\cU, \binom{\cU}{2}}$. Observe, that for each $\cP\in \cT$, $S_{\cP}$ is a triangle cut of $K_n$. This is because if $\cP$ separates $u,v\in \cU$, then for any $w\in \cU$, such that $u,v\neq w$ we have one of three following cases. Let $P_u, P_v, P_w$ be the parts of $\cP$ containing $u,v,w$, respectively. Then:
\begin{enumerate}
\item Assume $P_u=P_w\neq P_v$. Then $(u,v)\in S_{\cP}$ and $(u,w)\notin S_{\cP}$, so $(v,w)\in S_{\cP}$.
\item Assume $P_v=P_w\neq P_u$. Then $(u,v)\in S_{\cP}$ and $(v,w)\notin S_{\cP}$, so $(u,w)\in S_{\cP}$.
\item Assume $P_u\neq P_v$ and $P_u\neq P_w$. Then $(u,v)\in S_{\cP}$ and $(u,w)\in S_{\cP}$, so $(v,w)\in S_{\cP}$.
\end{enumerate}
For a visual example see Figure \ref{fig:partition_cutting}.

Let $\cS_{\cF}=\brc{S_{\cP}\colon\cP\in\cF}$. Applying Lemma~\ref{lemma:cutting_clique_upper_bound} to this set of triangle cuts gives the stated bound.
\end{proof}

Consider budget tested by the algorithm. If a tested budget $B$ fails, then $\spr{\cov\br{\cC_B}}<\alpha T$. The approximation guarantee implies $\OPT_{MC}\br{\cI_{MC}^{B}}<T$. Any separating subfamily of cost at most $B$ would, by Lemma~\ref{lemma:cutting_clique_lower_bound}, cover at least $T$ pairs, hence $B<\OPT_{SF}\br{\cI}$. Thus every failed budget is below $\OPT_{SF}\br{\cI}$. Therefore whenever binary search finds the lowest budget $B$ that passes, we get that $B\leq \OPT_{SF}\br{\cI}$. Applying the preceding lemma with $\gamma=\alpha$ bounds every part of $\cP_{\widehat{\cF}}$ by $\sqrt{1-\br{1-1/e}\cdot\br{1-f}}\cdot\spr{\cU}$, which gives the claimed size violation factor.
\end{proof}

Note, that by the guarantee established above, for the particular case  when $f=1/2$, required in our analysis, for each $P\in \cP_{\cF}$ we have $\spr{P}\leq \sqrt{\frac{e+1}{2\cdot e}}\cdot n$, and the size violation factor becomes
\[
\frac{\sqrt{1-\br{1-1/e}/2}}{1/2}=2\cdot\sqrt{\frac{e+1}{2\cdot e}}=\sqrt{\frac{2\cdot\br{e+1}}{e}}.
\]

\section{Putting Things Together}

Now we combine the results of the previous sections. Plugging in the appropriate values of $\alpha,\beta$ for $f=1/2$, obtained in Theorem~\ref{thm:separating_subfamily} into Theorem~\ref{thm:main_algorithm} yields the following corollary:
\begin{corollary}
    There exists a polynomial-time algorithm that constructs a decision tree $D$ such that
    \[
        \COSTW\br{D}\leq\frac{\log n}{\log\br{\sqrt{\frac{2\cdot e}{e+1}}}}\cdot \OPT^W\br{\cI}\leq 3.65\cdot \log n\cdot \OPT^W\br{\cI}
    \]
    \[
        \COSTA\br{D}\leq2\cdot\frac{\log n}{\log\br{\sqrt{\frac{2\cdot e}{e+1}}}}\cdot \OPT^A\br{\cI}\leq 7.3\cdot \log n\cdot \OPT^A\br{\cI}.
    \]
\end{corollary}
\section{New and Simple ODT Application: Adaptive Set Cover}

In this section, we provide a new, independent application of the \textsc{Optimal Decision Tree}: we give an $\cO\br{\log \br{k\cdot n}}$-approximation algorithm for a problem we call \textsc{Adaptive Set Cover}.
Note, that this is an independent result, however the paper's main focus is to simplify adaptive combinatorial optimization tasks. 
We explain the problem using the following motivating example.
Imagine, that we are given a system consisting of many components and a collection of possible failure scenarios. Each scenario specifies the set of components that fail if it occurs. These sets and a probability distribution over the scenarios are known in advance, while the realized scenario is hidden. Our goal is to identify the realized scenario and fix its faulty components. To do so, we can use a number of tests, each consisting of several components and an associated cost. Performing a test incurs a specified amount of cost, reveals every failed component in that test, and fixes all of them. The goal is to find an adaptive strategy that minimizes the expected cost required to fix the entire system. In particular this means that the tests are performed sequentially, and upon obtaining the information about the status of all components of the test, we are allowed to adaptively choose the next test to perform, based on the already gathered information.

We model this problem as follows: The components of the system are encoded by a universe $\cU$ of $n$ elements and the possible failure scenarios are represented by a family $\cR$ of $k$ (different) subsets of $\cU$ called \emph{realizations}, together with a weight function $w\colon\cR\to \mathbb{N}$\footnote{Note that if we were to rescale these weights in order to add up to $1$, we could reinterpret them as a probability distribution over $\cR$, however we use natural-numbered weights for the sake of convenience.}. The tests correspond to a family $\cS$ of $m$ covering subsets of $\cU$ (by which we mean that $\cU=\bigcup_{S\in \cS}S$), each with an associated cost $c\colon \cS \to \mathbb{N}$. We assume that among $\cR$ there is a \emph{hidden realization} $R^* \in \cR$ required to be covered.
For any set $S\in \cS$ and realization $R \in \cR$, we denote by $O_{S, R} = S\cap R$ the outcome of the test $S$ when the realization is $R$. 
Performing a set/test $S \in \cS$ incurs a cost of $c\br{S}$, reveals $O_{S, R^*}$ and simultaneously covers it. A realization $R'\in \cR$ is said to be \emph{consistent} with this response if $O_{S, R'}=O_{S, R^*}$.
A \emph{covering strategy} is a decision tree $D=\br{V, E}$, where each internal node is labeled with a set/test $S \in \cS$, and each edge $(u, v) \in E$ is associated with a possible outcome of a test, and by extension all of the realizations still consistent with it. Moreover, each leaf of $D$ corresponds to a realization $R \in \cR$, denoting the fact that the strategy has covered $R$ and identified it to be the hidden realization $R^*$\footnote{Here, we mention that whenever we have that $R^*\subseteq R'$ for some other realization $R'\in \cR$, then upon covering of $R^*$, $R'$ might still be consistent with all of the previous responses and then the strategy needs to pick additional sets in order to become aware that the entirety of $R^*$ is indeed covered.}.
 The cost of covering $R$ using $D$ is defined as 
\[
\cost\br{D, R} = \sum_{S \in \pathbt_D\br{R}} c\br{S}.
\]
The \emph{average-case covering cost} of the strategy $D$ is defined as 
\[
\COSTA\br{D} = \sum_{R \in \cR} w\br{R}\cdot \cost\br{D, R}.
\] 
The \emph{worst-case covering cost} of the strategy $D$ is defined as 
\[
\COSTW\br{D} = \max_{R \in \cR} \brc{\cost\br{D, R}}.
\]
The goal of the \textsc{Adaptive Set Cover} problem is to find a covering strategy $D$ that minimizes either $\COSTA$ or $\COSTW$.

\subsection{More preliminaries}

For any set $S\in \cS$, we denote by $\cO_S$ the set of all possible outcomes of the test $S$, i.e., $\cO_S = \brc{O_{S, R}: R \in \cR}$. By $\cP_{S}$ we denote a partition of $\cR$ into disjoint subsets, where each part corresponds to a possible outcome of the test $S$, i.e., $\cP_{S} = \brc{\brc{R \in \cR: O_{S, R} = O}: O \in \cO_S}$.   

We now define the following problems relevant to the \textsc{Adaptive Set Cover}:
\begin{enumerate}
    \item \textsc{Set Cover}: Given a universe $\cU$ of $n$ elements, a family $\cS$ of covering subsets of $\cU$ (where $\cU=\bigcup_{S\in \cS}S$), and a cost function $c: \cS \to \mathbb{R}_{\geq 0}$, the goal is to find a cover $\cC \subseteq \cS$ which minimizes the cost $c\br{\cC}=\sum_{S \in \cC} c(S)$, subject to $\cov\br{\cC} = \cU$. 
    \item \textsc{Realization Identification}: Here the input is the same as in \textsc{Adaptive Set Cover} and the goal is again to find a decision tree $D$. However, each leaf node of $D$ corresponding to a realization $R \in \cR$, denotes the fact that the strategy has identified $R$ to be $R^*$ (but not necessarily covered it), which means that $R$ is the only realization consistent with outcomes of the previously performed tests.
    The cost of identifying $R$ using $D$ is defined as $\cost\br{D, R} = \sum_{v \in \pathbt_D\br{R}} c\br{S_v}$. The average-case and worst-case covering costs of the strategy $D$ are defined as $\COSTA\br{D} = \sum_{R \in \cR} w\br{R}\cdot \cost\br{D, R}$ and $\COSTW\br{D} = \max_{R \in \cR}\brc{\cost\br{D, R}}$, respectively. The goal is to find a strategy $D$ that minimizes either $\COSTA$ or $\COSTW$.
\end{enumerate}

For an \textsc{Adaptive Set Cover} instance $\cI$, denote by $\OPT_{ASC}^{A}\br{\cI}$ and $\OPT_{ASC}^{W}\br{\cI}$ the optimum average-case and worst-case costs of the \textsc{Adaptive Set Cover} and by $\OPT_{RI}^{A}\br{\cI}$ and $\OPT_{RI}^{W}\br{\cI}$ the optimum average-case and worst-case costs of the \textsc{Realization Identification}. For a \textsc{Set Cover} instance $\cI$, denote by $\OPT_{SC}\br{\cI}$ its optimum cost. Observe that $\OPT_{RI}^{A}\br{\cI} \leq \OPT_{ASC}^{A}\br{\cI}$ and $\OPT_{RI}^{W}\br{\cI} \leq \OPT_{ASC}^{W}\br{\cI}$ since any strategy that solves the \textsc{Adaptive Set Cover} also solves the \textsc{Realization Identification}.
\subsection{The algorithm for the Adaptive Set Cover}

It is relatively straightforward to see that the \textsc{Adaptive Set Cover} is a generalization of the \textsc{Set Cover}. Observe that if we were to have an instance of the \textsc{Adaptive Set Cover} with a single realization $R=\cU$ with probability~$1$, then any feasible decision tree $D$ would consist of a sequence of sets/tests which cover~$R$. Moreover, both cost criteria have the same value, and optimizing any of them is equivalent to solving the \textsc{Set Cover}, since the cost of any decision tree becomes the cost of all picked sets/tests.
This means that achieving an approximation ratio beyond $\cO\br{\log n}$ is impossible unless $\textup{P}=\textup{NP}$ \cite{AnalyticalApproachToParallelRepetition}. We show that this characterization is essentially tight, namely that the \textsc{Adaptive Set Cover} admits an $\cO\br{\log \br{k\cdot n}}$-approximation algorithm. The same result can also be obtained by reducing the problem to the \textsc{Adaptive Submodular Ranking} \cite{AdaptiveSubmodularityANewApproachToActiveLearningAndStochasticOptimization}. However, this framework is very general and requires much more technicality and careful analysis. Here, we obtain the same guarantees for our setup while employing a very simplistic argument.

\begin{theorem}\label{thm:asc}
There exists an algorithm which achieves an $\cO\br{\log \br{k\cdot n}}$-approximation for both the average-case and worst-case versions of the \textsc{Adaptive Set Cover} simultaneously.
\end{theorem}
\begin{proof}
    The algorithm is roughly inspired by the work on the \textsc{Adaptive TSP} \cite{ApproximationAlgorithmsForOptimalDecisionTreesAndAdaptiveTSPProblems}, where the idea was to firstly identify the hidden scenario and then to find a TSP tour on the vertices of the scenario. 
    We proceed in a similar manner, firstly we identify the hidden realization $R^*\in \cR$ and then find a cover of $R^*$. The identification phase is executed by reducing the problem to the \textsc{Optimal Decision Tree}, while the covering phase is achieved using the \textsc{Set Cover} with the universe being the identified realization $R^*$. We note, that here, if the weights are uniform we can use our algorithm for the ODT, otherwise we can use the algorithm of \cite{DiagnosisDeterminationDecisionTreesOptimizingSimultaneouslyWorstAndExpectedTestingCost} which works for arbitrary weights. Let us stress that both of these algorithms are bicriteria: the single decision tree they output is simultaneously an $\cO\br{\log n}$-approximation with respect to the average-case and with respect to the worst-case cost. In the covering phase, we use the greedy heuristic for the \textsc{Set Cover} of \cite{AGreedyHeuristicForTheSetCoveringProblem}.
    We now describe the algorithm in detail.
    \begin{enumerate}
        \item Phase 1, Identification: Run a (bicriteria) $\cO\br{\log \spr{\cH}}$-approximation algorithm for the ODT on the following instance $\cI_{ODT}$: the set of hypotheses is $\cH = \cR$ with the weight function $w$. For each set $S \in \cS$ create a test $t_S=\cP_S$ with $c\br{t_S} = c\br{S}$. Let $D_{ODT}$ be the resulting decision tree, where each $t_S$ is replaced by the corresponding set $S$.
        \item Phase 2, Covering: For every $R \in \cR$, run an $\cO\br{\log \spr{\cU}}$ approximation algorithm for the \textsc{Set Cover} on the instance $\cI_R=\br{\cU=R, \cS, c}$. Let $\cC_R$ be the returned solution. Then, replace the leaf of $D_{ODT}$ corresponding to $R$ by a sequence of internal nodes corresponding to the sets in $\cC_R$ (in an arbitrary order). Below the last node of this sequence hang a leaf corresponding to $R$.
        Let $D$ be the resulting decision tree.
        \item Return $D$.
    \end{enumerate}
    
    \begin{lemma}\label{lem:asc_identification}
        The decision tree $D_{ODT}$ is a feasible solution for the \textsc{Realization Identification}.
    \end{lemma}
    \begin{proof}
     Consider the identification process when solving the \textsc{Realization Identification} problem on $\cI$. Whenever a set/test $S \in \cS$ is chosen, all elements in $O_{S,R^*}$ become identified. This means that all of the realizations $R'\in \cR$ such that $O_{S, R'} \neq O_{S, R^*}$ are eliminated from the set of possible realizations. Observe, that this is indeed also true for the ODT instance we constructed, since the the set of hypotheses is $\cH = \cR$ and each test $t_S$ partitions $\cH$ into disjoint subsets $\cP_{t_S}$ in the exact same way as the corresponding set $S \in \cS$ partitions $\cR$ into disjoint subsets $\cP_S$. In other words, $\cP_{t_S}=\cP_S$. This means that the \textsc{Realization Identification} is a special case of the ODT and therefore, the decision tree $D_{ODT}$ is indeed a feasible solution for it.
    \end{proof}

    By the above lemma, and the approximation of the algorithm used in Phase 1, we get
    \[
    \COSTA\br{D_{ODT}} \leq \cO\br{\log k}\cdot \OPT^A\br{\cI_{ODT}} = \cO\br{\log k} \cdot \OPT_{RI}^{A}\br{\cI}\leq \cO\br{\log k}\cdot \OPT_{ASC}^{A}\br{\cI},
    \]
    and analogously 
    \[
    \COSTW\br{D_{ODT}} \leq \cO\br{\log k}\cdot \OPT^W\br{\cI_{ODT}} = \cO\br{\log k}\cdot \OPT_{RI}^{W}\br{\cI} \leq \cO\br{\log k}\cdot \OPT_{ASC}^{W}\br{\cI}.
    \]

    Now, we observe that whenever we proceed according to $D_{ODT}$ while trying to cover the hidden realization $R^* \in \cR$, when the leaf corresponding to $R^*$ is reached, $R^*$ is guaranteed to be identified. Therefore, by choosing next a sequence of sets covering $R^*$ we obtain a feasible solution $D$ for the \textsc{Adaptive Set Cover}. We now argue about the cost of this strategy.

    \begin{lemma}\label{lem:asc_cost}
        For the decision tree $D$ we have 
        \[\COSTA\br{D}\leq \cO\br{\log \br{k\cdot n}}\cdot \OPT_{ASC}^{A}\br{\cI}
        \quad\text{and}\quad
        \COSTW\br{D}\leq \cO\br{\log \br{k\cdot n}}\cdot \OPT_{ASC}^{W}\br{\cI}.
        \]
    \end{lemma}
    \begin{proof}
    The average-case cost of $D$ can be decomposed into two parts: the cost incurred in Phase 1 (identification) which is at most $\cO\br{\log k} \cdot \OPT_{ASC}^{A}\br{\cI}$ and the cost incurred in Phase 2 (covering). We reason about the cost of Phase 2. Let $D_A^*$ be an optimal solution for the average-case \textsc{Adaptive Set Cover}. For every realization $R \in \cR$, observe that $\pathbt_{D_A^*}\br{R}$ consists of nodes which are associated with a subfamily of sets which is a feasible solution for the \textsc{Set Cover} instance $\cI_R=\br{R, \cS, c}$. This means that $\OPT_{SC}\br{\cI_R} \leq \cost\br{D_A^*, R}$ and as a consequence $c\br{\cC_R}\leq \cO\br{\log n}\cdot \cost\br{D_A^*, R}$. Therefore, we get that
    \begin{align*}
        \COSTA\br{D} &= \sum_{R \in \cR} w\br{R}\cdot \cost\br{D, R} 
        \\
        &\leq \sum_{R \in \cR} w\br{R}\cdot \cost\br{D_{ODT}, R} + \sum_{R \in \cR} w\br{R}\cdot c\br{\cC_R} \\
        &\leq \sum_{R \in \cR} w\br{R}\cdot \cost\br{D_{ODT}, R} + \cO\br{\log n}\cdot \sum_{R \in \cR} w\br{R}\cdot \cost\br{D_A^*, R} \\
        &= \COSTA\br{D_{ODT}} + \cO\br{\log n}\cdot \COSTA\br{D_A^*} 
        \\&\leq \cO\br{\log k}\cdot \OPT_{ASC}^{A}\br{\cI} + \cO\br{\log n}\cdot \OPT_{ASC}^{A}\br{\cI} 
        \\&= \cO\br{\log \br{k\cdot n}}\cdot \OPT_{ASC}^{A}\br{\cI}.
    \end{align*}
    where the first inequality bounds each path cost by the sum of its identification and covering costs, and the second follows from $c\br{\cC_R}\leq \cO\br{\log n}\cdot \cost\br{D_A^*, R}$. The third inequality uses the Phase 1 approximation bound and the optimality of $D_A^*$. The intermediate equalities are by the definition of $\COSTA$, and the final equality uses $\log k+\log n=\log\br{k\cdot n}$.

    Analogously, for a decision tree $D_W^*$, optimizing the worst-case \textsc{Adaptive Set Cover} 
    \begin{align*}
        \COSTW\br{D} &= \max_{R \in \cR} \brc{\cost\br{D, R}}
        \\
        &\leq \max_{R \in \cR} \brc{\cost\br{D_{ODT}, R}} + \max_{R \in \cR} \brc{c\br{\cC_R}} \\
        &\leq \COSTW\br{D_{ODT}} + \cO\br{\log n}\cdot \max_{R \in \cR} \brc{\cost\br{D_W^*, R}} \\
        &= \COSTW\br{D_{ODT}} + \cO\br{\log n}\cdot \COSTW\br{D_W^*} 
        \\&\leq \cO\br{\log k}\cdot \OPT_{ASC}^{W}\br{\cI} + \cO\br{\log n}\cdot \OPT_{ASC}^{W}\br{\cI} 
        \\&= \cO\br{\log \br{k\cdot n}}\cdot \OPT_{ASC}^{W}\br{\cI}.
    \end{align*}
    where the first inequality bounds the maximum total cost by the sum of the maximum identification and covering costs. The second follows since the sets on each path of $D_W^*$ cover the corresponding realization $R$, so $c\br{\cC_R}\leq \cO\br{\log n}\cdot \cost\br{D_W^*, R}$. The third inequality uses the Phase 1 approximation bound and the optimality of $D_W^*$. The intermediate equalities are by the definition of $\COSTW$, and the final equality again uses $\log k+\log n=\log\br{k\cdot n}$, as required.
    \end{proof}

    Combining Lemma~\ref{lem:asc_identification} and \ref{lem:asc_cost} concludes the proof of Theorem~\ref{thm:asc}.
\end{proof}
\section{Conclusions}
We have presented a simple greedy algorithm for the \textsc{Optimal Decision Tree} problem with simultaneous approximation factors below $3.65\log_2 n$ for worst-case cost and $7.3\log_2 n$ for average-case cost. Although simultaneous $\cO\br{\log n}$-approximation for both objectives was already known~\cite{DiagnosisDeterminationDecisionTreesOptimizingSimultaneouslyWorstAndExpectedTestingCost}, our result is siginificantly easier to understand. Stronger guarantees are known when optimizing only one criterion: $4\cdot\ln n$ for average-case cost with uniform probabilities and non-uniform costs~\cite{MinimumCostAdaptiveSubmodularCover}, and $\log n$ for worst-case cost with non-uniform costs~\cite{TheCostComplexityOfInteractiveLearning}. These single-objective procedures can be transformed into a bicriteria strategy~\cite{TradingOffWorstAndExpectedCostInDecisionTreeProblems}, but no matter the choice of parameters in this construction, one cannot beat both of our simultaneous guarantees at once.

We also introduce the \textsc{Adaptive Set Cover} formulation and use our ODT algorithm to obtain an $\cO\br{\log \br{k\cdot n}}$-approximation for it.
The advantage of our simple approach is that the dependence on both $k$ and $n$ in the approximation ratio is transparent: one logarithmic factor comes from identifying the hidden realization, while the other occurs since we need to cover all elements of the identified realization. Note, that if there is an upper bound on the maximum size of a realization, call it $r$, then the approximation ratio improves to $\cO\br{\log \br{k\cdot r}}$.
We leave as an open question whether one can improve this guarantee down to $\cO\br{\log r}$, or alternatively establish an inapproximability result excluding such a possibility.

Currently, our algorithm for the \textsc{Optimal Decision Tree} does not work for arbitrary probability distribution/weight function over the hypotheses, where the aim is to minimize the weighted average-case cost. 
A relatively easy argument certifying the approximation guarantee for this setting was presented in~\cite{ASimpleApproximationAlgorithmForOptimalDecisionTree}, however the analysis of their algorithm is still significantly more involved than the one presented in this work.
It is interesting whether a similarly simple approach together with an appropriate analysis can be designed for this setting as well. Note, that even if there existed a bicriteria constant-factor approximation algorithm for the generalization of the \textsc{Separating Subfamily} problem, where each element has a weight, our approach would only yield an $O\br{\log \br{1/p_{\min}}}$-approximation for the average-case cost ODT, where $p_{\min}$ is the minimum probability among all of the hypotheses.

One could also ask, whether it is possible to give a true, single-criteria, constant-factor approximation for the \textsc{Separating Subfamily} which is an interesting problem by itself. We answer in the negative below, by showing that the problem generalizes the \textsc{Set Cover}. A hardness result can also be provided by reducing from the \textsc{Maximum Coverage}, which shows that obtaining a $\br{1, 1+1/e-\epsilon}$-approximation guarantee is hard as well.

\begin{theorem}
  There is no $\br{\br{1-o\br{1}}\cdot\ln\br{n/2}, 1}$ nor $\br{1, 1+1/e-\epsilon}$-approximation algorithm, for the \textsc{Separating Subfamily} for any $\epsilon>0$, unless $\textup{P}=\textup{NP}$.
\end{theorem}
\begin{proof}
Let $\cI_{SC}=\br{\cU, \cS}$ be a uniform-cost \textsc{Set Cover} instance. We construct a \textsc{Separating Subfamily} instance $\cI_{SF}=\br{\cU', \cT, c, f}$ as follows. We set $\cU'=\cU\cup [\spr{\cU}]$.
For each $u\in \cU'$ we create a partition $\cP_u=\brc{\brc{u}, \cU'\setminus \{u\}}$ with $c\br{\cP_u}=\infty$. For each $S\in \cS$ we create a partition $\cP_S=\brc{S, \cU'\setminus S}$ with $c\br{\cP_S}=1$. We set $f=1/2$. Since for any $u, v\in \cU'$ there exists $\cP\in \cT$ separating $u$ and $v$, the instance is feasible. 
We claim that any solution to the \textsc{Separating Subfamily} instance $\cI_{SF}$ corresponds to a solution to the \textsc{Set Cover} instance $\cI_{SC}$ with equal cost, and vice versa.

Firstly, we assume that for any sensible solution $\cF$ to the \textsc{Separating Subfamily} instance $\cI_{SF}$, $\cF\subseteq \brc{\cP_S\colon S\in \cS}$, otherwise $c\br{\cF}=\infty$. Since no other $\cP\in \cT$ can separate members of $[\spr{\cU}]$ amongst themselves, $\cP_{\cF}$ must separate all elements in $\cU$ from $[\spr{\cU}]$. To see that the corresponding family of sets $\cC=\brc{S\in \cS\colon \cP_S\in \cF}$ is a cover of $\cU$, observe that if
there exists $u\in \cU$ such that $u\notin \cov\br{\cC}$, then $u$ is not separated from $[\spr{\cU}]$ by~$\cP_{\cF}$, a contradiction. 
Conversely, for any $S\in \cS$ and $u\in S$, let $u\in P_u^S\in \cP_S$. Then, $\spr{P_u^S}\leq\spr{\cU}=\spr{\cU'}/2$. This means that $\cF$ is a feasible solution for the $\cI_{SF}$ instance, if $\cC$ is a cover. This is because by definition of a cover and the previous discussion, for every $u\in \cU$, where $u\in P_u'\in \cP_{\cF}$, we have $P_u'\subseteq P_u^S$ for some $S\in \cC$. Therefore, $\spr{P_u'}\leq \spr{P_u}\leq \spr{\cU}=\spr{\cU'}/2$ and $\cF$ is a feasible solution for the \textsc{Separating Subfamily} instance $\cI_{SF}$.
Since the costs of both corresponding solutions are equal and $\spr{\cU'}=2\cdot \spr{\cU}$, if the \textsc{Separating Subfamily} admits an algorithm with an approximation of
\[
\br{1-o\br{1}}\cdot \ln \br{n/2}=\br{1-o\br{1}}\cdot \ln \br{2\cdot\spr{\cU}/2} = \br{1-o\br{1}}\cdot \ln \spr{\cU}
\]
we immediately obtain an $\br{1-o\br{1}}\cdot \ln \spr{\cU}$-approximation for the \textsc{Set Cover}.

The reduction from a uniform-cost \textsc{Maximum Coverage} instance $\cI_{MC}=\brc{\cU, \cS, B}$ to the \textsc{Separating Subfamily} instance $\cI_{SF}=\br{\cU', \cT, c, 1/2}$ is exactly the same construction. Here we rely on the fact that the $\br{1-1/e}$-inapproximability of the \textsc{Maximum Coverage} established in \cite{TheBudgetedMaximumCoverageProblem} already holds on the family of instances in which the optimum covers the whole universe, i.e.
\[
\OPT_{MC}\br{\cI_{MC}}=\spr{\cU},
\]
since these instances are produced from \textsc{Set Cover} instances in which one has to distinguish whether $B$ sets suffice to cover $\cU$, or no $B$ sets cover more than a $\br{1-1/e+\epsilon}$ fraction of $\cU$.

Let $\cC^*$ be an optimal solution for $\cI_{MC}$, so that $\cov\br{\cC^*}=\cU$ and $c\br{\cC^*}\leq B$, and let $\cF^*=\brc{\cP_S\colon S\in \cC^*}$ be the corresponding subfamily. By the discussion above, every $u\in \cU$ belongs to a part $P_u\in \cP_{\cF^*}$ with $\spr{P_u}\leq \spr{\cU}=\spr{\cU'}/2$, and the only remaining part of $\cP_{\cF^*}$ is $\cU'\setminus \cov\br{\cC^*}=[\spr{\cU}]$, which is of size $\spr{\cU'}/2$ as well. Therefore, $\cF^*$ is a feasible solution for $\cI_{SF}$ and $\OPT_{SF}\br{\cI_{SF}}\leq c\br{\cF^*}\leq B$.

Let $\cF$ be the solution returned by a $\br{1, \beta}$-approximation algorithm for the \textsc{Separating Subfamily} on $\cI_{SF}$, and let $\cC=\brc{S\in \cS\colon \cP_S\in \cF}$ be the corresponding cover. Since $c\br{\cC}=c\br{\cF}\leq \OPT_{SF}\br{\cI_{SF}}\leq B$, the family $\cC$ is a feasible solution for $\cI_{MC}$. Let $P\in \cP_{\cF}$ be the part containing $[\spr{\cU}]$. As every other part of $\cP_{\cF}$ is a subset of $\cov\br{\cC}\subseteq \cU$, the part $P$ is the largest one. Consequently, the violation guarantee of the $\br{1,\beta}$-approximation gives
\[
\spr{P}\leq \beta\cdot \frac{\spr{\cU'}}{2}=\beta\cdot \spr{\cU}.
\]
On the other hand, $\cov\br{\cC}=\cU'\setminus P$, and hence
\[
\spr{\cov\br{\cC}}=\spr{\cU'}-\spr{P}\geq 2\cdot \spr{\cU}-\beta\cdot \spr{\cU}=\br{2-\beta}\cdot \OPT_{MC}\br{\cI_{MC}}.
\]

Thus, we get an $2-\beta$-approximate solution for the \textsc{Maximum Coverage}. 
Assume, that we have an $\br{1, \beta = 1+1/e-\epsilon}$-approximation for the \textsc{Separating Subfamily} for some $\epsilon>0$. Then, we get an $2-\beta=1-1/e+\epsilon$-approximation for the \textsc{Maximum Coverage}.

Since both of the guarantees provided by the above reductions are impossible to obtain unless $\Pclass=\NP$ \cite{TheBudgetedMaximumCoverageProblem,AnalyticalApproachToParallelRepetition}, the claim follows.
\end{proof}

\bibliographystyle{plain}
\bibliography{references}

\end{document}